\documentclass[reprint,amsmath,amssymb,aps,pra,floatfix]{revtex4-2}

\usepackage{amsmath,amssymb}
\usepackage{graphicx}
\usepackage{dcolumn}
\usepackage{booktabs}
\usepackage{bm}
\usepackage{array}
\usepackage[hidelinks]{hyperref}

\usepackage{amsthm}
\usepackage{braket}

\newtheorem{theorem}{Theorem}
\newtheorem{proposition}[theorem]{Proposition}

\theoremstyle{definition}
\newtheorem{definition}[theorem]{Definition}
\theoremstyle{remark}

\newcommand{\la}{\lambda}
\newcommand{\si}{\sigma}
\newcommand{\ov}[1]{\overline{#1}}
\newcommand{\SH}{\mathbb{H}_{s}}
\newcommand{\Rey}{\operatorname{Re}}
\newcommand{\Hy}{\operatorname{Hy}}
\newcommand{\Norm}[1]{\mathcal{N}(#1)}
\newcommand{\Lag}{\mathcal{L}}
\newcommand{\sqi}{\iota}
\newcommand{\ii}{\mathrm{i}}
\newcommand{\NormQ}[1]{\mathcal{N}_{\!Q}(#1)}

\begin{document}

\title{Split-quaternionic structure and canonical quantization of the Bateman dual oscillator}

\author{Mathieu Beau}
\affiliation{Department of Physics, University of Massachusetts Boston, Boston, MA 02125, USA}

\date{\today}

\begin{abstract}
Bateman's dual oscillator embeds a damped harmonic oscillator and an independent anti-damped adjoint coordinate in a conservative doubled system. We develop a split-quaternionic formulation that remains consistent from the classical equations through ordinary complex canonical quantization. The doubled dynamics is written as a two-sided hypercomplex evolution whose conserved indefinite quadratic form is the split norm, while the sign of the generator square unifies the underdamped, overdamped, and critical regimes. The Bateman Lagrangian is recovered as a real scalar action on the split-complex configuration subalgebra. After quantization, operator-valued split quaternions belong to the complexified algebra $\mathbb C\otimes_{\mathbb R}\mathbb H_s\simeq M_2(\mathbb C)$. Noncommuting mechanical components require a symmetrized quantum split norm, which exactly reproduces the symmetrically ordered canonical Hamiltonian. Null idempotents resolve the crossed canonical structure of the doubled phase space. In particular, damping yields $[m\dot{\widehat y},m\dot{\widehat x}]=\mathrm{i}\hbar m\gamma$ and a Robertson--Schr\"odinger uncertainty relation. The limit $\gamma\to0$ retains the doubled positive/negative geometry; the ordinary single-oscillator Heisenberg relation is recovered only after selecting the positive normal mode, or by undoubling before quantization. The scalar split action gives the exact quadratic propagator and Gaussian evolution. A formal trace over the auxiliary coordinate is finally compared with high-temperature Caldeira--Leggett dynamics, clarifying that the doubled model is a coherent zero-noise embedding rather than a microscopic reservoir theory.
\end{abstract}
\keywords{Bateman oscillator; split quaternions; canonical quantization; quantum dissipation; indefinite metric; path integral}

\maketitle

\section{Introduction}

The damped harmonic oscillator,
\begin{equation}\label{eq:osc}
\ddot x+\gamma\dot x+\omega_0^2 x=0,
\end{equation}
has no conserved mechanical energy for $\gamma>0$, but it admits conservative enlargements. Bateman introduced a dual completion in which an independent adjoint coordinate obeys the time-reversed anti-damped equation \cite{Bateman1931}. Feshbach and Tikochinsky subsequently developed the canonical quantization of this doubled system \cite{FeshbachTikochinsky1977}. Related approaches include the Caldirola--Kanai model \cite{Caldirola1941,Kanai1948}, inverse-variational constructions \cite{Chandrasekar2007,Bersani2021,Riewe1996}, doubled-variable methods \cite{MorseFeshbach1953,MartinezPerez2018}, and nonconservative variational principles \cite{Galley2013}. Reviews of quantum damped oscillators include Refs.~\cite{Dekker1981,UmYeonGeorge2002}.

Split quaternions were introduced by Cockle \cite{Cockle1849}; their indefinite geometry and matrix realizations are standard \cite{OzdemirErgin2006,KulaYayli2007}. Their use in dissipative mechanics, including the algebraic classification of damping regimes, has also been explored \cite{Legrand2022}. Connections among Bateman dynamics, magnetic analogies, and noncommutative structures were studied by Pal, Nandi, and Chakraborty in a Moyal-deformed setting \cite{PalNandiChakraborty2018}. The present construction is different: no noncommutative configuration space is postulated. The relevant noncommutativity arises intrinsically between covariant mechanical momenta after ordinary canonical quantization and is organized by split-quaternion conjugation.

The doubled Lagrangian used below is linearly equivalent to Bateman's, and the corresponding quantum theory belongs to the Feshbach--Tikochinsky class. We therefore claim neither a new variational embedding nor a new quantization prescription. The contribution is a unified algebraic organization of the known doubled system, together with consequences that become transparent in that organization. First, a two-sided split-quaternionic equation simultaneously encodes the damped and adjoint sectors, the conserved indefinite norm, and the three damping regimes. Second, operator-valued split quaternions in $\mathbb C\otimes_{\mathbb R}\mathbb H_s$ require a symmetrized quantum norm that reproduces the canonical Hamiltonian. Third, the null-idempotent decomposition resolves the crossed canonical pairs and yields a damping-controlled Robertson--Schr\"odinger relation for the two mechanical velocities. Fourth, the Bateman action and exact quadratic propagator can be written directly in split-complex scalar form.

A final comparison with the high-temperature Markovian Caldeira--Leggett model is included only to delimit the physical meaning of the doubled construction. Tracing the adjoint coordinate is a state-dependent formal reduction; it does not generate the diffusion, decoherence, or thermalization associated with a microscopic bath.

\section{Split-quaternion algebra and calculus}

\subsection{Algebra, conjugation, and norm}

\begin{definition}[Split-quaternion algebra]
\label{def:splitalgebra}
A split quaternion is
\begin{equation}
z=a+b\la+c\sqi+d\si,
\qquad a,b,c,d\in\mathbb R,
\end{equation}
with basis relations
\begin{equation}\label{eq:basis}
\la^2=\si^2=1,\qquad \sqi^2=-1,\qquad \sqi\la\si=1,
\end{equation}
which imply
\begin{equation}\label{eq:prods}
\sqi\la=\si=-\la\sqi,\qquad
\si\la=\sqi=-\la\si,\qquad
\si\sqi=\la=-\sqi\si.
\end{equation}
We write $\Rey z=a$ and $\Hy z=b\la+c\sqi+d\si$. Split conjugation and the split norm are
\begin{equation}
\ov{z}=a-b\la-c\sqi-d\si,
\qquad
\Norm{z}=z\ov{z}=a^2-b^2+c^2-d^2.
\label{eq:norm}
\end{equation}
The norm is a real indefinite quadratic form of signature $(+,-,+,-)$, not a positive modulus.
\end{definition}

The bilinear product is
\begin{align}\label{eq:fullprod}
z_1 z_2 ={}& (a_1a_2+b_1b_2-c_1c_2+d_1d_2)\nonumber\\
&+(a_1b_2+b_1a_2-c_1d_2+d_1c_2)\,\la\nonumber\\
&+(a_1c_2+c_1a_2-b_1d_2+d_1b_2)\,\sqi\nonumber\\
&+(a_1d_2+d_1a_2-b_1c_2+c_1b_2)\,\si.
\end{align}

\begin{proposition}[Faithful matrix realization]
\label{prop:matrixrealization}
The map
\begin{equation}
\mathfrak m(z)=
\begin{pmatrix}
a+b&c-d\\
-c-d&a-b
\end{pmatrix}
\end{equation}
is an algebra isomorphism $\SH\cong M_2(\mathbb R)$. Under this map,
\begin{equation}
\det \mathfrak m(z)=\Norm{z},
\end{equation}
and split conjugation reverses products:
$\ov{z_1z_2}=\ov{z_2}\,\ov{z_1}$.
\end{proposition}

\begin{proof}
Direct multiplication of the two matrix images reproduces Eq.~\eqref{eq:fullprod}; the determinant is Eq.~\eqref{eq:norm}. The conjugation identity follows either from Eq.~\eqref{eq:fullprod} or from the corresponding adjugate relation in $M_2(\mathbb R)$.
\end{proof}

\begin{definition}[Null idempotents and split bilinear forms]
\label{def:nullidempotents}
The elements
\begin{equation}
e_+=\frac{1+\la}{2},\qquad e_-=\frac{1-\la}{2}
\label{eq:nullidempotents}
\end{equation}
satisfy $e_\pm^2=e_\pm$, $e_+e_-=0$, and $\ov{e_\pm}=e_\mp$. On the commutative split-complex subalgebra
\begin{equation}
\mathbb D:=\operatorname{span}_{\mathbb R}\{1,\la\}\subset\SH,
\end{equation}
define
\begin{equation}
\langle \xi,\eta\rangle_s:=\Rey\!\left(\xi\ov{\eta}\right),
\qquad
\omega_s(\xi,\eta):=\Rey\!\left(\la\xi\ov{\eta}\right).
\label{eq:splitbilinearforms}
\end{equation}
Then $\langle \xi,\xi\rangle_s=\Norm{\xi}$, whereas $\omega_s$ is antisymmetric and distinguishes the two null sectors.
\end{definition}

Throughout the classical construction, $\sqi$ denotes the noncentral split-quaternion generator. After quantization, the ordinary complex unit is denoted by $\ii$ and commutes with $\la$, $\sqi$, and $\si$. Operator-valued split quaternions therefore belong to $\mathbb C\otimes_{\mathbb R}\SH\simeq M_2(\mathbb C)$, not to a quaternionic Hilbert-space formulation.

\subsection{Exponential trichotomy}

\begin{proposition}[Exponential trichotomy]
\label{prop:exptrichotomy}
Let $A=a_0+h$, where $a_0=\Rey A$ and $h=\Hy A$. Since
\begin{equation}
h^2=b^2-c^2+d^2\in\mathbb R,
\end{equation}
the exponential is
\begin{equation}\label{eq:exp}
e^{At}=e^{a_0t}
\begin{cases}
\displaystyle \cos(\kappa t)+\frac{h}{\kappa}\sin(\kappa t),
& h^2<0,\; \kappa=\sqrt{-h^2},\\[1.3ex]
\displaystyle \cosh(\eta t)+\frac{h}{\eta}\sinh(\eta t),
& h^2>0,\; \eta=\sqrt{h^2},\\[1.3ex]
1+ht,&h^2=0.
\end{cases}
\end{equation}
Moreover,
\begin{equation}\label{eq:modexp}
\Norm{e^{At}}=e^{2\Rey(A)t}.
\end{equation}
\end{proposition}

\begin{proof}
The mixed terms in $h^2$ cancel, so the power series separates into even and odd powers of a real scalar. Equation~\eqref{eq:modexp} follows from $[A,\ov{A}]=0$ and $A+\ov{A}=2\Rey A$.
\end{proof}

For the oscillator generator introduced below, the elliptic, hyperbolic, and parabolic cases in Proposition~\ref{prop:exptrichotomy} coincide with underdamping, overdamping, and critical damping.

\subsection{Two-sided evolution and norm conservation}

\begin{theorem}[Norm law for two-sided evolution]
\label{thm:twosidednorm}
Let
\begin{equation}\label{eq:twosided}
\dot z=Az+zB,
\qquad A,B\in\SH
\end{equation}
with constant generators. Then
\begin{equation}
z(t)=e^{At}z(0)e^{Bt}
\end{equation}
and
\begin{equation}\label{eq:dmod}
\frac{d}{dt}\Norm{z}
=2\bigl(\Rey A+\Rey B\bigr)\Norm{z}.
\end{equation}
Consequently, the norm is conserved for every initial datum if and only if $\Rey A+\Rey B=0$.
\end{theorem}

\begin{proof}
Differentiate $z\ov{z}$, use $\ov{Az+zB}=\ov{z}\,\ov{A}+\ov{B}\,\ov{z}$, and note that $z\ov{z}$ is real and central. This gives Eq.~\eqref{eq:dmod}. The stated condition is then necessary and sufficient for conservation for all initial data.
\end{proof}

\section{Time-reversible doubling of the damped oscillator}

\begin{definition}[Bateman dual pair]
\label{def:batemanpair}
The Bateman dual completion of Eq.~\eqref{eq:osc} consists of an independent adjoint coordinate $y$ satisfying
\begin{equation}\label{eq:adjoint}
\ddot y-\gamma\dot y+\omega_0^2y=0.
\end{equation}
Time reversal exchanges the two sectors,
\begin{equation}
\mathcal T:\quad x(t)\mapsto y(-t),\qquad y(t)\mapsto x(-t).
\end{equation}
The auxiliary coordinate is independent; it is not constrained to equal $x(-t)$.
\end{definition}

\begin{proposition}[Time-reversal-adapted normal coordinates]
\label{prop:normalcoordinates}
Define
\begin{equation}\label{eq:rhor}
\rho=\frac{x+y}{\sqrt2},\qquad r=\frac{x-y}{\sqrt2},
\end{equation}
with $u=\dot\rho$, $w=\dot r$, and $\alpha_d=m\gamma$. Then the Bateman pair is equivalent to
\begin{equation}\label{eq:trev}
m\dot u=-\alpha_d w-m\omega_0^2\rho,
\qquad
m\dot w=-\alpha_d u-m\omega_0^2r.
\end{equation}
Under $\mathcal T$, $\rho$ is even, $r$ is odd, $u$ is odd, and $w$ is even; hence Eq.~\eqref{eq:trev} is time-reversal invariant.
\end{proposition}

\begin{proof}
Substitute Eq.~\eqref{eq:rhor} and its inverse into Eqs.~\eqref{eq:osc} and \eqref{eq:adjoint}. The parity statements follow directly from the exchange $x(t)\leftrightarrow y(-t)$.
\end{proof}

The apparent coupling in $(\rho,r)$ is only a linear rewriting of the decoupled damped and adjoint equations in $(x,y)$; it is not a microscopic reservoir interaction.

\section{Split-quaternion encoding and conserved norm-energy}

\begin{theorem}[Split-quaternionic normal form of Bateman dynamics]
\label{thm:sqnormalform}
Introduce
\begin{equation}
p=m(u+w\la),\qquad q=\rho+\la r,
\end{equation}
and
\begin{equation}\label{eq:z}
z=p+\sqi m\omega_0q
=mu+mw\la+m\omega_0\rho\,\sqi+m\omega_0r\,\si.
\end{equation}
Then Eq.~\eqref{eq:trev} is equivalent to
\begin{equation}\label{eq:zeq}
\dot z=Az+zB,
\qquad
A=-\frac{\gamma}{2}\la+\omega_0\sqi,
\qquad
B=-\frac{\gamma}{2}\la.
\end{equation}
The generator satisfies
\begin{equation}
A^2=\frac{\gamma^2}{4}-\omega_0^2,
\end{equation}
so its sign gives the underdamped, overdamped, and critical regimes. Moreover, $\Rey A=\Rey B=0$, and therefore
\begin{equation}\label{eq:E}
E:=\frac{\Norm{z}}{2m}
=\frac{m}{2}\left(u^2-w^2+\omega_0^2\rho^2-\omega_0^2r^2\right)
\end{equation}
is conserved.
\end{theorem}

\begin{proof}
Expand $Az+zB$ in the basis $1,\la,\sqi,\si$ and compare with Eq.~\eqref{eq:trev}; the component identities are given explicitly in Appendix~\ref{app:alg}. The classification follows from Proposition~\ref{prop:exptrichotomy}, and norm conservation follows from Theorem~\ref{thm:twosidednorm}.
\end{proof}

In the underdamped regime define
\begin{equation}\label{eq:Omega}
\Omega=\sqrt{\omega_0^2-\frac{\gamma^2}{4}}.
\end{equation}
The evolution factors are
\begin{align}
e^{At}&=\cos(\Omega t)+
\left(-\frac{\gamma}{2\Omega}\la+\frac{\omega_0}{\Omega}\sqi\right)\sin(\Omega t),\label{eq:ea}\\
e^{Bt}&=\cosh\!\left(\frac{\gamma t}{2}\right)
-\la\sinh\!\left(\frac{\gamma t}{2}\right).\label{eq:eb}
\end{align}
Equation~\eqref{eq:E} identifies the conserved Bateman quadratic form with the split norm of a single doubled dynamical variable. Its compact harmonic form must nevertheless be distinguished from the canonical Hamiltonian obtained by Legendre transformation.

\section{Lagrangian and canonical Hamiltonian}

In terms of $p$ and $q$, the conserved form is
\begin{equation}\label{eq:Hnorm}
E=\frac{\Norm{p}}{2m}+\frac{m\omega_0^2}{2}\Norm{q}.
\end{equation}
The cross terms in $\Norm{z}$ cancel because $\sqi q\ov{p}-p\ov{q}\sqi=0$, a consequence of $\sqi\la=-\la\sqi$. Equation~\eqref{eq:Hnorm} is a norm-energy written in mechanical variables; it is not yet the canonical Hamiltonian on phase space.

\begin{proposition}[Variational and Hamiltonian realization]
\label{prop:variational}
The Lagrangian
\begin{equation}\label{eq:Lcorr}
\Lag=\frac{m}{2}(u^2-w^2)
-\frac{m\omega_0^2}{2}(\rho^2-r^2)
+\frac{\alpha_d}{2}(ur-w\rho)
\end{equation}
reproduces Eq.~\eqref{eq:trev}. Its canonical momenta are
\begin{equation}\label{eq:canonicalmomenta}
\pi_\rho=mu+\frac{\alpha_d}{2}r,
\qquad
\pi_r=-mw-\frac{\alpha_d}{2}\rho,
\end{equation}
and its Legendre transform is
\begin{align}\label{eq:Hcan}
H_{\rm can}={}&
\frac{1}{2m}\left(\pi_\rho-\frac{\alpha_d}{2}r\right)^2
-\frac{1}{2m}\left(\pi_r+\frac{\alpha_d}{2}\rho\right)^2\nonumber\\
&+\frac{m\omega_0^2}{2}(\rho^2-r^2).
\end{align}
After using Eq.~\eqref{eq:canonicalmomenta}, $H_{\rm can}=E$.
\end{proposition}

\begin{proof}
The Euler--Lagrange equations and Legendre transform are evaluated in Appendix~\ref{app:alg}. Direct substitution of Eq.~\eqref{eq:canonicalmomenta} into Eq.~\eqref{eq:Hcan} gives Eq.~\eqref{eq:E}.
\end{proof}

For $\omega_0\to0$, Eq.~\eqref{eq:Lcorr} describes a particle on a two-dimensional configuration space of signature $(+,-)$ with a magnetic-type velocity coupling \cite{MartinezPerez2018}.

\section{Quantization}\label{sec:quantization}

\subsection{Canonical quantization of the doubled system}

Canonical quantization is performed on the doubled phase space $(\rho,r;\pi_\rho,\pi_r)$ by promoting these variables to operators and imposing
\begin{equation}
[\widehat\rho,\widehat\pi_\rho]=\ii\hbar,\qquad [\widehat r,\widehat\pi_r]=\ii\hbar,
\end{equation}
with the remaining equal-time commutators vanishing, and by promoting Eq.~\eqref{eq:Hcan} to a symmetrically ordered operator. This is a quantization of the conservative doubled system, not of Eq.~\eqref{eq:osc} alone. The operator is indefinite and has no lower-bounded ground-state spectrum, but an appropriate self-adjoint realization generates unitary evolution on the enlarged Hilbert space \cite{FeshbachTikochinsky1977,Dekker1981,ChruscinskiJurkowski2006}. 

After a standard linear transformation, the Hamiltonian can be written in terms of conjugate oscillator modes and its $SU(1,1)$ structure becomes explicit \cite{CeleghiniRasettiVitiello1992,ChruscinskiJurkowski2006}. The enlarged evolution should not be confused with the reduced dynamics generated by tracing a physical reservoir. The split-quaternionic formulation does not replace this standard complex quantization; instead, it packages the resulting mechanical operators and their indefinite quadratic invariant in a form adapted to the doubled geometry.

\subsection{Operator-valued split quaternions and the symmetrized quantum norm}
\label{sec:quantumsplit}

Let
\begin{equation}\label{eq:quantumqp}
\widehat q=\widehat\rho+\la\widehat r,
\qquad
\widehat p=\widehat P_u+\la\widehat P_w,
\end{equation}
where
\begin{equation}\label{eq:quantummechanicalcomponents}
\widehat P_u
=\widehat\pi_\rho-\frac{m\gamma}{2}\widehat r,
\qquad
\widehat P_w
=-\widehat\pi_r-\frac{m\gamma}{2}\widehat\rho.
\end{equation}
These are operator-valued elements of
$\mathbb C\otimes_{\mathbb R}\SH\simeq M_2(\mathbb C)$.
The bar $\ov{\phantom p}$ continues to denote conjugation in the split-quaternion factor only,
\begin{equation}
\ov{\widehat q}=\widehat\rho-\la\widehat r,
\qquad
\ov{\widehat p}=\widehat P_u-\la\widehat P_w,
\end{equation}
and should not be confused with the Hilbert-space adjoint $\dagger$.
The ordinary complex unit $\ii$ is central in this tensor-product algebra.

\begin{proposition}[Quantum split-norm structure]
\label{prop:quantumsplit}
The operator-valued split quaternions in Eq.~\eqref{eq:quantumqp} satisfy
\begin{align}
[\widehat q,\widehat p]&=0,
&[\widehat q,\ov{\widehat p}]&=2\ii\hbar,
\label{eq:qpcommutators}\\
[\widehat p,\ov{\widehat p}]&=-2\ii\hbar m\gamma\,\la.
\label{eq:psplitcommutator}
\end{align}
Because $\widehat p\ov{\widehat p}$ is not a scalar in the split-quaternion factor when $\gamma\neq0$, define the symmetrized quantum split norm by
\begin{equation}\label{eq:quantumnormdef}
\NormQ{\widehat a}
:=\frac12\left(
\widehat a\,\ov{\widehat a}
+\ov{\widehat a}\,\widehat a
\right).
\end{equation}
Then
\begin{equation}\label{eq:quantumnorms}
\NormQ{\widehat p}
=\widehat P_u^2-\widehat P_w^2,
\qquad
\NormQ{\widehat q}
=\widehat\rho^2-\widehat r^2,
\end{equation}
and the symmetrically ordered canonical Hamiltonian is
\begin{equation}\label{eq:Hquantumsplit}
\boxed{
\widehat H_{\rm can}
=\frac{\NormQ{\widehat p}}{2m}
+\frac{m\omega_0^2}{2}\NormQ{\widehat q}.}
\end{equation}
\end{proposition}

\begin{proof}
The canonical commutators imply
$[\widehat P_u,\widehat P_w]=\ii\hbar m\gamma$.
Since the split generators commute with the Hilbert-space operators,
Eq.~\eqref{eq:qpcommutators} follows by expanding the products and using
$\la^2=1$. Moreover,
\begin{align}
\widehat p\,\ov{\widehat p}
&=\widehat P_u^2-\widehat P_w^2
-\ii\hbar m\gamma\,\la,\\
\ov{\widehat p}\,\widehat p
&=\widehat P_u^2-\widehat P_w^2
+\ii\hbar m\gamma\,\la,
\end{align}
which gives Eqs.~\eqref{eq:psplitcommutator} and \eqref{eq:quantumnorms}.
The position components commute, so no analogous correction appears for
$\widehat q$. Substitution of Eq.~\eqref{eq:quantummechanicalcomponents}
into Eq.~\eqref{eq:Hquantumsplit} reproduces the symmetrically ordered
operator obtained from Eq.~\eqref{eq:Hcan}.
\end{proof}

Equation~\eqref{eq:Hquantumsplit} is the precise sense in which the classical split-quaternion norm survives canonical quantization. The unsymmetrized product acquires a damping-dependent split-vector term, while symmetrization extracts the scalar indefinite quadratic form. This construction remains within ordinary complex quantum mechanics and does not posit a quaternionic probability amplitude or a new quantization rule.

\subsection{Null-sector decomposition, cross-canonical structure, and uncertainty}
\label{sec:crosscanonical}

The idempotents in Definition~\ref{def:nullidempotents} make the physical meaning of Proposition~\ref{prop:quantumsplit} explicit. Introduce the Bateman coordinates
\begin{equation}\label{eq:xyrhotransform}
x=\frac{\rho+r}{\sqrt2},\qquad y=\frac{\rho-r}{\sqrt2}.
\end{equation}
Upon quantization, the same linear transformation is applied to the corresponding operators, so that the operator-valued split position decomposes as
\begin{equation}\label{eq:qnullsectors}
\widehat q
=\sqrt2\left(\widehat x e_++\widehat y e_-\right),
\qquad
\ov{\widehat q}
=\sqrt2\left(\widehat x e_-+\widehat y e_+\right).
\end{equation}
The canonical one-form is preserved by
\begin{equation}
\widehat\pi_x=\frac{\widehat\pi_\rho+\widehat\pi_r}{\sqrt2},
\qquad
\widehat\pi_y=\frac{\widehat\pi_\rho-\widehat\pi_r}{\sqrt2},
\end{equation}
and Eq.~\eqref{eq:canonicalmomenta} gives
\begin{equation}\label{eq:pixy}
\widehat\pi_x=m\dot{\widehat y}-\frac{m\gamma}{2}\widehat y,
\qquad
\widehat\pi_y=m\dot{\widehat x}+\frac{m\gamma}{2}\widehat x.
\end{equation}
Define the covariant mechanical momenta
\begin{equation}\label{eq:Kxy}
\widehat K_x:=\widehat\pi_x+\frac{m\gamma}{2}\widehat y=m\dot{\widehat y},
\qquad
\widehat K_y:=\widehat\pi_y-\frac{m\gamma}{2}\widehat x=m\dot{\widehat x}.
\end{equation}
Here $\dot{\widehat x}$ and $\dot{\widehat y}$ denote Heisenberg time derivatives generated by $\widehat H_{\rm can}$.
Then the split mechanical momentum in Eq.~\eqref{eq:quantumqp} is
\begin{equation}\label{eq:pnullsectors}
\widehat p
=\sqrt2\left(\widehat K_y e_++\widehat K_x e_-\right),
\qquad
\ov{\widehat p}
=\sqrt2\left(\widehat K_y e_-+\widehat K_x e_+\right).
\end{equation}
Thus split conjugation exchanges the two null sectors. In particular, $\widehat q$ and $\widehat p$ align each coordinate with its own mechanical velocity and commute, whereas $\widehat q$ and $\ov{\widehat p}$ align each coordinate with the opposite-sector mechanical velocity and form the canonical pair of Proposition~\ref{prop:quantumsplit}.

\begin{proposition}[Null-sector commutators and damping-induced uncertainty]
\label{prop:nullsectorcommutators}
On a common invariant domain, Eqs.~\eqref{eq:qnullsectors}--\eqref{eq:pnullsectors} and Proposition~\ref{prop:quantumsplit} imply
\begin{align}
[\widehat x,\widehat K_y]&=0,
&[\widehat x,\widehat K_x]&=\ii\hbar,\label{eq:crosscommx}\\
[\widehat y,\widehat K_x]&=0,
&[\widehat y,\widehat K_y]&=\ii\hbar,\label{eq:crosscommy}
\end{align}
and
\begin{equation}\label{eq:velocitycomm}
[\widehat K_x,\widehat K_y]=\ii\hbar m\gamma.
\end{equation}
Equivalently,
\begin{equation}\label{eq:pu-pw-comm}
[\widehat P_u,\widehat P_w]=\ii\hbar m\gamma.
\end{equation}
Every normalized state with finite second moments therefore satisfies
\begin{equation}\label{eq:velocityRS}
\operatorname{Var}(\widehat K_x)\operatorname{Var}(\widehat K_y)
-\operatorname{Cov}_{\rm sym}(\widehat K_x,\widehat K_y)^2
\geq\frac{\hbar^2m^2\gamma^2}{4},
\end{equation}
where
\begin{equation}
\operatorname{Cov}_{\rm sym}(\widehat X,\widehat Y)
=\frac12\left\langle\left\{\widehat X-\langle\widehat X\rangle,\widehat Y-\langle\widehat Y\rangle\right\}\right\rangle.
\end{equation}
In velocity variables,
\begin{equation}\label{eq:velocityRSscaled}
\operatorname{Var}(\dot{\widehat y})\operatorname{Var}(\dot{\widehat x})
-\operatorname{Cov}_{\rm sym}(\dot{\widehat y},\dot{\widehat x})^2
\geq\frac{\hbar^2\gamma^2}{4m^2}.
\end{equation}
\end{proposition}

\begin{proof}
Insert Eqs.~\eqref{eq:qnullsectors} and \eqref{eq:pnullsectors} into
$[\widehat q,\widehat p]=0$ and
$[\widehat q,\ov{\widehat p}]=2\ii\hbar$. Orthogonality of $e_+$ and $e_-$ yields Eqs.~\eqref{eq:crosscommx} and \eqref{eq:crosscommy}. Likewise, Eq.~\eqref{eq:psplitcommutator} is equivalent to Eq.~\eqref{eq:velocitycomm}; the relation between $(\widehat K_x,\widehat K_y)$ and $(\widehat P_u,\widehat P_w)$ then gives Eq.~\eqref{eq:pu-pw-comm}. The remaining inequalities are the Robertson--Schr\"odinger relation.
\end{proof}

The damping term has a geometric expression already in the split variables. On $\mathbb D$, the real one-form
\begin{equation}\label{eq:splitconnection}
\mathcal A_{\rm sq}
=-\frac{m\gamma}{2}\Rey\!\left(\la\ov{q}\,dq\right)
\end{equation}
becomes
\begin{equation}
\mathcal A_{\rm sq}
=-\frac{m\gamma}{2}y\,dx
+\frac{m\gamma}{2}x\,dy.
\end{equation}
Its curvature is the constant two-form
\begin{equation}
d\mathcal A_{\rm sq}=m\gamma\,dx\wedge dy.
\end{equation}
Equation~\eqref{eq:psplitcommutator} is therefore the split-quaternion packaging of the familiar noncommutativity of covariant momenta in a constant magnetic-type field. Related Bateman--noncommutativity correspondences have been obtained by embedding the model in a Moyal-deformed phase space \cite{PalNandiChakraborty2018}. Here the configuration operators remain commuting; the noncommutativity is instead the intrinsic commutator of the covariant mechanical momenta generated by the velocity coupling. The off-diagonal kinetic metric lowers the coordinate index by exchanging sectors, which is why $\widehat K_x=m\dot{\widehat y}$ and $\widehat K_y=m\dot{\widehat x}$.

\subsubsection{Conservative limit and recovery of the ordinary oscillator}

When $\gamma=0$, Eq.~\eqref{eq:psplitcommutator} gives
$[\widehat p,\ov{\widehat p}]=0$, but the canonical relation
$[\widehat q,\ov{\widehat p}]=2\ii\hbar$ remains. The ordinary positive and negative normal modes are recovered by scalar and $\lambda$ projections of the split variables:
\begin{align}
\widehat Q&:=\frac{\widehat q+\ov{\widehat q}}{2}=\widehat\rho,
&\widehat P&:=\frac{\widehat p+\ov{\widehat p}}{2}=\widehat P_u=\widehat\pi_\rho,\label{eq:positiveprojection}\\
\widehat R&:=\frac{\widehat q-\ov{\widehat q}}{2\la}=\widehat r,
&\widehat\Pi_R&:=-\frac{\widehat p-\ov{\widehat p}}{2\la}=\widehat\pi_r.\label{eq:negativeprojection}
\end{align}
Consequently,
\begin{equation}
[\widehat Q,\widehat P]=\ii\hbar,
\qquad
[\widehat R,\widehat\Pi_R]=\ii\hbar,
\end{equation}
and the Hamiltonian becomes
\begin{equation}\label{eq:Hgamma0}
\widehat H_0
=\left(\frac{\widehat P^2}{2m}+\frac{m\omega_0^2\widehat Q^2}{2}\right)
-\left(\frac{\widehat\Pi_R^2}{2m}+\frac{m\omega_0^2\widehat R^2}{2}\right).
\end{equation}
Both projected pairs satisfy their ordinary Robertson--Schr\"odinger inequalities. In particular,
\begin{equation}\label{eq:ordinaryHeisenberg}
\Delta\widehat Q\,\Delta\widehat P\geq\frac{\hbar}{2}.
\end{equation}
Selecting the positive scalar sector $(\widehat Q,\widehat P)$ recovers the ordinary harmonic oscillator. This is a change from the full doubled theory to its positive normal-mode subsystem, not the sharp imposition of $\widehat R=\widehat\Pi_R=0$ within the doubled Hilbert space.

By contrast, the Bateman coordinate operator $\widehat x$ is the coefficient of the null idempotent $e_+$ in Eq.~\eqref{eq:qnullsectors}, not the scalar projection $\widehat Q$. Therefore
\begin{equation}
[\widehat x,m\dot{\widehat x}]=[\widehat x,\widehat K_y]=0
\end{equation}
continues to hold at $\gamma=0$ in the unreduced doubled model. If the symbol $x$ is to denote an autonomous ordinary oscillator, one must undouble before quantization, starting from
\begin{equation}
\Lag_{\rm osc}[x]
=\frac{m}{2}\dot x^2-\frac{m\omega_0^2}{2}x^2,
\end{equation}
which, after quantization, gives $\widehat p_x=m\dot{\widehat x}$ and $[\widehat x,\widehat p_x]=\ii\hbar$.

Finally, the Heisenberg operator $\widehat x(t)$ depends linearly on the mutually commuting initial null-sector operators $\widehat x(0)$ and $\widehat K_y(0)=m\dot{\widehat x}(0)$. Hence
\begin{equation}\label{eq:commutinghistory}
[\widehat x(t),\widehat x(t')]=0
\end{equation}
for all times for which the doubled evolution exists. This self-nondemolition property belongs to Bateman's enlarged embedding and should not be interpreted as an autonomous quantization of the damped oscillator.

\subsection{Split-quaternionic action and exact propagator}
\label{sec:sqpathintegral}

The path integral can be formulated directly on the split-complex configuration subalgebra $\mathbb D$. Let
\begin{equation}
q(t)=\rho(t)+\la r(t)\in\mathbb D.
\end{equation}

\begin{proposition}[Scalar split-quaternionic action]
\label{prop:sqaction}
The real scalar Lagrangian
\begin{equation}\label{eq:Lsq}
\Lag_{\rm sq}[q,\dot q]
=\frac{m}{2}\Norm{\dot q}
-\frac{m\omega_0^2}{2}\Norm{q}
-\frac{m\gamma}{2}\Rey\!\left(\la\ov{q}\,\dot q\right)
\end{equation}
is exactly Eq.~\eqref{eq:Lcorr}. In the null-idempotent decomposition
\begin{equation}
q=\sqrt2(xe_++ye_-),
\end{equation}
it becomes the symmetric Bateman Lagrangian
\begin{equation}\label{eq:LxyFull}
\Lag_{\rm sq}
=m\dot x\dot y
+\frac{m\gamma}{2}(x\dot y-y\dot x)
-m\omega_0^2xy.
\end{equation}
\end{proposition}

\begin{proof}
For $q=\rho+\la r$,
$\Norm{\dot q}=u^2-w^2$, $\Norm{q}=\rho^2-r^2$, and
$\Rey(\la\ov{q}\dot q)=\rho w-r u$. This gives Eq.~\eqref{eq:Lcorr}. Using Eq.~\eqref{eq:xyrhotransform} then gives Eq.~\eqref{eq:LxyFull}.
\end{proof}

The scalar action and its path integral are
\begin{align}
S_{\rm sq}[q]
&=\int_0^t\Lag_{\rm sq}[q,\dot q] \,ds,\nonumber\\
K(q_f,q_i;t)
&=\int_{q(0)=q_i}^{q(t)=q_f}\!\mathcal Dq\,
\exp\!\left(\frac{\ii}{\hbar}S_{\rm sq}[q]\right).
\label{eq:sqpathintegral}
\end{align}
where $\mathcal Dq=\mathcal D\rho\,\mathcal Dr=\mathcal Dx\,\mathcal Dy$. Thus the amplitudes remain ordinary complex amplitudes; the split quaternion organizes the real doubled configuration space and its scalar action.

Using the underdamped frequency $\Omega$ defined in Eq.~\eqref{eq:Omega}, the endpoint solutions in the two null sectors are
\begin{align}
x(s)={}&e^{-\gamma s/2}\Biggl[
 x_i\cos(\Omega s)\nonumber\\
&\hspace{1.7cm}+
\frac{e^{\gamma t/2}x_f-x_i\cos(\Omega t)}{\sin(\Omega t)}
\sin(\Omega s)\Biggr],\label{eq:xclharm}\\
y(s)={}&e^{\gamma s/2}\Biggl[
 y_i\cos(\Omega s)\nonumber\\
&\hspace{1.7cm}+
\frac{e^{-\gamma t/2}y_f-y_i\cos(\Omega t)}{\sin(\Omega t)}
\sin(\Omega s)\Biggr].\label{eq:yclharm}
\end{align}

\begin{theorem}[Exact quadratic propagator in split coordinates]
\label{thm:sqpropagator}
Away from the caustics $\sin(\Omega t)=0$, define
\begin{align}
A_\Omega(t)&=m\Omega\cot(\Omega t),\label{eq:Aomega}\\
B_\Omega(t)&=-\frac{m\Omega e^{-\gamma t/2}}{\sin(\Omega t)},\label{eq:Bomega}\\
C_\Omega(t)&=-\frac{m\Omega e^{\gamma t/2}}{\sin(\Omega t)}.\label{eq:Comega}
\end{align}
For $q_i,q_f\in\mathbb D$, where
$\langle q_i,q_f\rangle_s=\Rey\!\left(q_i\ov{q_f}\right)$ and
$\omega_s(q_i,q_f)=\Rey\!\left(\la q_i\ov{q_f}\right)$ as in Eq.~\eqref{eq:splitbilinearforms}, the on-shell action is
\begin{align}\label{eq:SclSQ}
S_{\rm cl}^{\rm sq}(q_f,q_i;t)
={}&\frac{A_\Omega}{2}\left(\Norm{q_i}+\Norm{q_f}\right)\nonumber\\
&+\frac{B_\Omega+C_\Omega}{2}\langle q_i,q_f\rangle_s\nonumber\\
&+\frac{B_\Omega-C_\Omega}{2}\omega_s(q_i,q_f).
\end{align}
Equivalently, in null components,
\begin{equation}\label{eq:SclMain}
S_{\rm cl}
=A_\Omega(x_i y_i+x_f y_f)
+B_\Omega x_i y_f
+C_\Omega x_f y_i.
\end{equation}
On the first caustic interval, the exact propagator is
\begin{equation}\label{eq:KMain}
K(q_f,q_i;t)
=\frac{m\Omega}{2\pi\hbar\sin(\Omega t)}
\exp\!\left(\frac{\ii}{\hbar}S_{\rm cl}^{\rm sq}(q_f,q_i;t)\right),
\end{equation}
with continuation across later caustics by the usual Maslov phase.
\end{theorem}

\begin{proof}
On a classical path, Eq.~\eqref{eq:LxyFull} satisfies
\begin{equation}
\Lag_{\rm sq}
=\frac{d}{dt}\left[\frac{m}{2}(x\dot y+y\dot x)\right].
\end{equation}
Inserting Eqs.~\eqref{eq:xclharm} and \eqref{eq:yclharm} gives Eq.~\eqref{eq:SclMain}. The identities
\begin{align}
\Norm{q}&=2xy,\\
\langle q_i,q_f\rangle_s&=x_i y_f+y_i x_f,\\
\omega_s(q_i,q_f)&=x_i y_f-y_i x_f
\end{align}
then give Eq.~\eqref{eq:SclSQ}. The Van Vleck determinant yields the prefactor in Eq.~\eqref{eq:KMain}; details and the short-time normalization are given in Appendix~\ref{app:pi}.
\end{proof}

The overdamped and critical cases follow by analytic continuation of Proposition~\ref{prop:exptrichotomy}. In the free damped limit $\omega_0\to0$, $\Omega\to\ii\gamma/2$ and the same split action \eqref{eq:SclSQ} applies with
\begin{equation}
A_\Omega\to\frac{a_\gamma(t)}{2},\qquad
B_\Omega\to\frac{b_\gamma(t)}{2},\qquad
C_\Omega\to\frac{c_\gamma(t)}{2},
\end{equation}
where
\begin{align}
a_\gamma(t)&=m\gamma\coth\!\left(\frac{\gamma t}{2}\right),\\
b_\gamma(t)&=\frac{2m\gamma}{1-e^{\gamma t}},\\
c_\gamma(t)&=\frac{2m\gamma e^{\gamma t}}{1-e^{\gamma t}}.
\end{align}
The corresponding prefactor is
\begin{equation}
\frac{m\gamma}{4\pi\hbar\sinh(\gamma t/2)}.
\end{equation}

\subsection{Evolution of a doubled Gaussian state}\label{sec:doubledgaussian}

For explicit Gaussian integration, we now use the faithful two-component representation of the null-sector coefficients of $q$. The split-coordinate endpoint action \eqref{eq:SclSQ}, equivalently Eq.~\eqref{eq:SclMain}, can be written in matrix form. Let
\begin{equation}
\mathbf X_i=\begin{pmatrix}x_i\\y_i\end{pmatrix},\qquad
\mathbf X_f=\begin{pmatrix}x_f\\y_f\end{pmatrix}.
\end{equation}
Then
\begin{equation}\label{eq:SMatrix}
S_{\rm cl}
=
\frac12 \mathbf X_i^{T}\mathsf{Q}_i \mathbf X_i
+
\mathbf X_i^{T}\mathsf{B} \mathbf X_f
+
\frac12 \mathbf X_f^{T}\mathsf{Q}_f \mathbf X_f,
\end{equation}
where
\begin{equation}\label{eq:QBmatrices}
\mathsf{Q}_i=\mathsf{Q}_f=
\begin{pmatrix}
0&A_\Omega(t)\\
A_\Omega(t)&0
\end{pmatrix},
\qquad
\mathsf{B}=
\begin{pmatrix}
0&B_\Omega(t)\\
C_\Omega(t)&0
\end{pmatrix}.
\end{equation}
In the free limit these entries reduce to $A_\Omega\to a_\gamma(t)/2$, $B_\Omega\to b_\gamma(t)/2$, and $C_\Omega\to c_\gamma(t)/2$.
For illustration, choose the factorized isotropic Gaussian on the doubled configuration space with width parameter $\ell>0$,
\begin{equation}\label{eq:initialDoubledGaussian}
\Psi_0(\mathbf X_i)
=
\frac{1}{\sqrt{2\pi\ell^2}}
\exp\!\left[-\frac{1}{4\ell^2}\mathbf X_i^T\mathbf X_i\right].
\end{equation}
In this two-component representation, write $K(\mathbf X_f,\mathbf X_i;t)$ for the same kernel. The evolved wave function is
\begin{equation}
\Psi_t(\mathbf X_f)=\int_{\mathbb R^2}K(\mathbf X_f,\mathbf X_i;t)\Psi_0(\mathbf X_i)\,d^2\mathbf X_i.
\end{equation}
Using \eqref{eq:KMain} and \eqref{eq:SMatrix}, this is an elementary complex Gaussian integral. Define
\begin{equation}\label{eq:AJDoubled}
\mathcal A=\frac{1}{2\ell^2}\mathbb{I}-\frac{\ii}{\hbar}\mathsf{Q}_i,
\qquad
\mathbf J=\frac{\ii}{\hbar}\mathsf{B}\mathbf X_f,
\end{equation}
with the branch of $\sqrt{\det \mathcal A}$ chosen by continuity from $t=0$. Then
\begin{align}\label{eq:GaussianEvolvedDoubled}
\Psi_t(\mathbf X_f)
={}&
\frac{m\Omega}{2\pi\hbar\sin(\Omega t)}
\frac{1}{\sqrt{2\pi\ell^2}}
\frac{2\pi}{\sqrt{\det \mathcal A}}
\nonumber\\
&\times
\exp\!\left[
\frac{\ii}{2\hbar}\mathbf X_f^T\mathsf{Q}_f \mathbf X_f
+
\frac12 \mathbf J^T \mathcal A^{-1}\mathbf J
\right].
\end{align}
Thus the chosen doubled Gaussian remains Gaussian under the exact time-symmetric propagator. More general doubled Gaussian states, including unequal widths and $x$--$y$ correlations, evolve in the same way but lead to different reduced covariances. Neither the initial state of the auxiliary sector nor a reduction prescription is fixed by the classical damped equation. For the illustrative comparison below, we trace over $y$.

\subsection{Reduced position density of the doubled model}\label{sec:reduceddensity}

For the specified doubled state, define the formal reduced density matrix of the physical coordinate by
\begin{equation}\label{eq:rhoSQreduced}
\rho_{\rm SQ}(x,x';t)
=
\int_{-\infty}^{\infty}dy\,
\Psi_t(x,y)\Psi_t^*(x',y).
\end{equation}
Its diagonal gives the position density
\begin{equation}\label{eq:rhoSQdiag}
\rho_{\rm SQ}(x,t)
=
\rho_{\rm SQ}(x,x;t)
=
\int_{-\infty}^{\infty}dy\,|\Psi_t(x,y)|^2.
\end{equation}
For the Gaussian state \eqref{eq:GaussianEvolvedDoubled}, write
\begin{equation}
\Psi_t(\mathbf X_f)=\mathcal N_t\exp\left[-\frac12 \mathbf X_f^T G_t \mathbf X_f\right],
\qquad \mathbf X_f=(x_f,y_f)^T.
\end{equation}
With the matrix convention of Eq.~\eqref{eq:SMatrix}, the complex width matrix is
\begin{equation}\label{eq:GtDoubled}
G_t
=
-\frac{\ii}{\hbar}\mathsf{Q}_f
+
\frac{1}{\hbar^2}
\mathsf{B}^T\left(\frac{1}{2\ell^2}\mathbb{I}-\frac{\ii}{\hbar}\mathsf{Q}_i\right)^{-1}\mathsf{B}.
\end{equation}
Let the ordinary complex real part be taken entrywise, and define
\begin{equation}
R_t=\operatorname{Re}_{\mathbb C}(G_t)
=
\begin{pmatrix}
R_{xx}(t)&R_{xy}(t)\\
R_{xy}(t)&R_{yy}(t)
\end{pmatrix}.
\end{equation}
Then
\begin{equation}
|\Psi_t(x,y)|^2
=
|\mathcal N_t|^2
\exp\left[-
\begin{pmatrix}x&y\end{pmatrix}
R_t
\begin{pmatrix}x\\y\end{pmatrix}
\right].
\end{equation}
The Gaussian integral over $y$ gives
\begin{equation}\label{eq:rhoSQGaussian}
\boxed{
\rho_{\rm SQ}(x,t)
=
\sqrt{\frac{\kappa_{\rm SQ}(t)}{\pi}}
\exp\left[-\kappa_{\rm SQ}(t)x^2\right]
}
\end{equation}
with
\begin{equation}\label{eq:kappaSQ}
\boxed{
\kappa_{\rm SQ}(t)
=
R_{xx}(t)-\frac{R_{xy}^2(t)}{R_{yy}(t)}
}
\end{equation}
and hence the variance of the Gaussian is
\begin{equation}\label{eq:varianceSQ}
\boxed{
\Sigma_{xx}^{\rm SQ}(t)=\frac{1}{2\kappa_{\rm SQ}(t)}.
}
\end{equation}
For a non-centered Gaussian the same expression holds after replacing $x$ by $x-\bar x_{\rm SQ}(t)$, where $\bar x_{\rm SQ}(t)$ follows the classical damped trajectory. The resulting width depends on the auxiliary-sector covariance and is therefore not determined by the damped equation alone.

\section{Reduced-state comparison with Caldeira--Leggett dynamics}\label{sec:open}

The formal reduction in Eq.~\eqref{eq:rhoSQreduced} should not be assigned the meaning of a bath trace. In the doubled construction, $y$ is an adjoint coordinate introduced to restore a conservative time-symmetric dynamics, and the reduced state depends on its chosen initial covariance. In the Caldeira--Leggett model, environmental variables are instead eliminated after a microscopic system--bath coupling, producing noise and diffusion \cite{CaldeiraLeggett1983,BreuerPetruccione2002,Weiss2012}.

For a concise Gaussian comparison, consider the high-temperature Markovian Caldeira--Leggett covariance matrix
\begin{equation}
V(t)=
\begin{pmatrix}
\Sigma_{xx}^{\rm CL}(t)&\Sigma_{xp}^{\rm CL}(t)\\
\Sigma_{xp}^{\rm CL}(t)&\Sigma_{pp}^{\rm CL}(t)
\end{pmatrix},
\end{equation}
which obeys
\begin{equation}\label{eq:CLcovariance}
\dot V=A_{\rm CL}V+VA_{\rm CL}^T+D_{\rm CL},
\end{equation}
with
\begin{equation}\label{eq:CLmatrices}
A_{\rm CL}=
\begin{pmatrix}
0&1/m\\
-m\omega_0^2&-\gamma
\end{pmatrix},
\qquad
D_{\rm CL}=
\begin{pmatrix}
0&0\\
0&2m\gamma k_BT
\end{pmatrix}.
\end{equation}
We use
\begin{equation}\label{eq:CLinitial}
V(0)=\operatorname{diag}\!\left(\ell^2,\frac{\hbar^2}{4\ell^2}\right),
\end{equation}
matching the initial physical-coordinate marginal of Eq.~\eqref{eq:initialDoubledGaussian}.

Figures~\ref{fig:varianceHarmonic} and \ref{fig:varianceFree} compare the formal doubled variance \eqref{eq:varianceSQ} with Eq.~\eqref{eq:CLcovariance}. In the harmonic case, the doubled variance follows coherent enlarged-system evolution, whereas the Caldeira--Leggett variance approaches the thermal scale fixed by diffusion and damping. In the free case, the doubled variance saturates for the specified factorized auxiliary state, while the Caldeira--Leggett variance continues to grow diffusively. These curves illustrate a structural difference, not an auxiliary-state-independent prediction of the doubled model. The full quantum Brownian model also contains zero-temperature bath fluctuations and generally non-Markovian memory; the split-quaternionic construction is therefore not a $T=0$ Caldeira--Leggett limit.

\begin{figure}[t]
\centering
\includegraphics[width=\columnwidth]{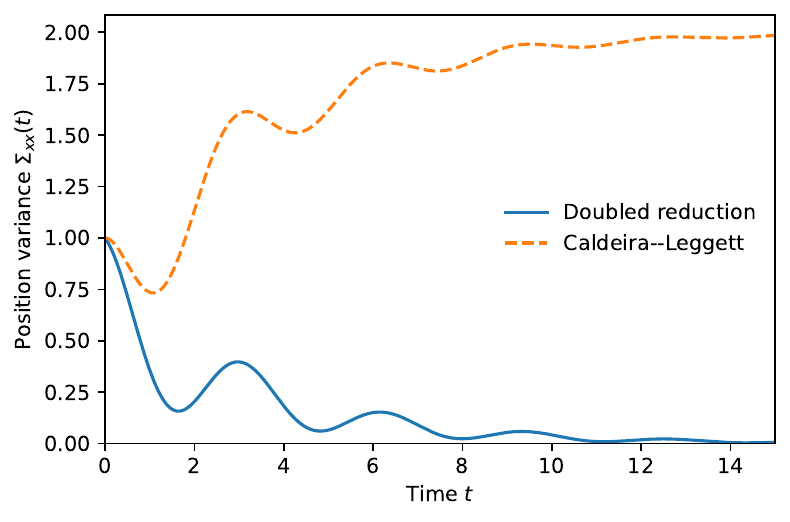}
\caption{Reduced position variance for the specified factorized doubled Gaussian and for the high-temperature Markovian Caldeira--Leggett model in the underdamped harmonic case. Parameters are $m=\hbar=1$, $\omega_0=1$, $\gamma=0.3$, $\ell=1$, and $k_BT=2$. The doubled result is a coherent, auxiliary-state-dependent reduction; the Caldeira--Leggett curve contains bath-induced diffusion.}
\label{fig:varianceHarmonic}
\end{figure}

\begin{figure}[t]
\centering
\includegraphics[width=\columnwidth]{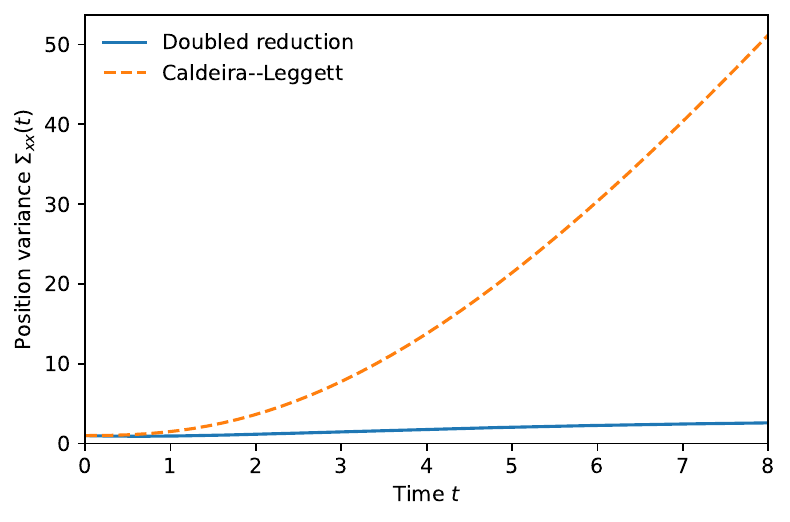}
\caption{Same comparison in the free damped limit $\omega_0=0$, with the remaining parameters as in Fig.~\ref{fig:varianceHarmonic}. For the chosen doubled state the reduced variance saturates, whereas the Caldeira--Leggett variance grows because of momentum diffusion.}
\label{fig:varianceFree}
\end{figure}

\begin{table*}[t]
\caption{Conceptual distinction between the doubled split-quaternion model and Caldeira--Leggett open-system dynamics.}
\label{tab:SQCL}
\renewcommand{\arraystretch}{1.15}
\setlength{\tabcolsep}{5pt}
\begin{tabular}{lll}
\toprule
\parbox[t]{2.5cm}{\textbf{Aspect}}
& \parbox[t]{5.3cm}{\textbf{Doubled split-quaternion model}}
& \parbox[t]{5.3cm}{\textbf{Caldeira--Leggett model}}\\
\midrule
\parbox[t]{2.5cm}{Structure}
& \parbox[t]{5.3cm}{Conservative enlarged system with physical and adjoint sectors}
& \parbox[t]{5.3cm}{Physical system coupled to and reduced from a reservoir}\\[0.7ex]
\midrule
\parbox[t]{2.5cm}{Damping and fluctuations}
& \parbox[t]{5.3cm}{Kinematic damping/anti-damping pair; no fundamental diffusion term}
& \parbox[t]{5.3cm}{Dissipation accompanied by bath noise and diffusion}\\[0.7ex]
\midrule
\parbox[t]{2.5cm}{Time symmetry}
& \parbox[t]{5.3cm}{Manifest in the enlarged dynamics}
& \parbox[t]{5.3cm}{Broken in the reduced effective dynamics}\\[0.7ex]
\midrule
\parbox[t]{2.5cm}{Temperature}
& \parbox[t]{5.3cm}{No intrinsic temperature or fluctuation--dissipation relation}
& \parbox[t]{5.3cm}{Finite- or zero-temperature reservoir dynamics}\\[0.7ex]
\midrule
\parbox[t]{2.5cm}{Main role}
& \parbox[t]{5.3cm}{Algebraic organization of reversible embedding and indefinite invariant}
& \parbox[t]{5.3cm}{Microscopic or effective description of decoherence and thermalization}\\[0.7ex]
\midrule
\parbox[t]{2.5cm}{Main limitation}
& \parbox[t]{5.3cm}{Auxiliary state and reduction are not fixed by the damped equation}
& \parbox[t]{5.3cm}{Results depend on bath spectrum and approximation scheme}\\
\bottomrule
\end{tabular}
\end{table*}

The two descriptions therefore answer different questions. The split-quaternion formalism clarifies the reversible doubled geometry; it is not a substitute for open-system theory.

\section{Conclusion}

We have given a split-quaternionic formulation of the Bateman dual oscillator that remains coherent across classical dynamics, canonical quantization, and the quadratic path integral. Classically, a two-sided hypercomplex equation encodes the damped and adjoint sectors, the conserved indefinite norm, and the underdamped--critical--overdamped trichotomy. The corresponding real scalar action on the split-complex subalgebra is exactly Bateman's Lagrangian in null components.

The quantum construction stays within ordinary complex Hilbert-space quantum mechanics. Operator-valued split quaternions lie in $\mathbb C\otimes_{\mathbb R}\mathbb H_s$, and the noncommutativity of their mechanical components makes the unsymmetrized product $\widehat p\,\overline{\widehat p}$ split-vector valued. Symmetrization defines a scalar quantum split norm that reproduces the canonical Hamiltonian. The null idempotents then resolve the compact relations into the crossed canonical pairs of the Bateman phase space. The velocity coupling generates
\begin{equation}
[m\dot{\widehat y},m\dot{\widehat x}]=\mathrm{i}\hbar m\gamma,
\end{equation}
and hence a damping-controlled Robertson--Schr\"odinger relation. This result is related to magnetic and noncommutative descriptions of Bateman dynamics, but here it arises without postulating a noncommutative configuration space.

The conservative limit must be interpreted carefully. Setting $\gamma=0$ leaves a doubled positive/negative oscillator pair, so the null-coordinate operator $\widehat x=(\widehat\rho+\widehat r)/\sqrt2$ still satisfies $[\widehat x,m\dot{\widehat x}]=0$. The standard Heisenberg relation is recovered for the positive normal mode $(\widehat\rho,\widehat\pi_\rho)$, or by removing the auxiliary sector before quantization. This distinction clarifies both the usefulness and the limitation of the doubled construction.

Finally, the scalar split action yields the exact propagator and Gaussian evolution. A formal trace over the adjoint coordinate produces a reduced density only after an auxiliary initial state is specified. Its contrast with Caldeira--Leggett dynamics confirms that the Bateman model is a coherent zero-noise embedding of damping, not a microscopic theory of diffusion, decoherence, or thermalization.

%\begin{acknowledgments}

%\end{acknowledgments}

\appendix

\section{Algebraic derivations}\label{app:alg}

We use the faithful representation $\SH\cong M_2(\mathbb{R})$,
\begin{equation}
1\mapsto\mathbb{I},\; \la\mapsto\begin{pmatrix}1&0\\0&-1\end{pmatrix},\; \sqi\mapsto\begin{pmatrix}0&1\\-1&0\end{pmatrix},\; \si\mapsto\begin{pmatrix}0&-1\\-1&0\end{pmatrix},
\end{equation}
so that $z=a+b\la+c\sqi+d\si\mapsto\big(\begin{smallmatrix}a+b & c-d\\ -c-d & a-b\end{smallmatrix}\big)$. One checks $\la^2=\si^2=\mathbb{I}$, $\sqi^2=-\mathbb{I}$, $\sqi\la\si=\mathbb{I}$, and $\det=a^2-b^2+c^2-d^2=\Norm{z}$.

\emph{Product \eqref{eq:fullprod}.} Expanding $z_1z_2$ with the bilinear extension of \eqref{eq:basis}--\eqref{eq:prods} and collecting by basis element gives the four components of \eqref{eq:fullprod}; equivalently, multiply the two $2\times2$ images and read off $(a,b,c,d)$ from the entries. Setting $z_2=\ov{z_1}$ collapses this to $\Norm{z}=a^2-b^2+c^2-d^2$.

\emph{Exponential trichotomy \eqref{eq:exp}.} For $A=a_0+h$ with $h=\Hy A$, the hypercomplex part squares to a real scalar, $h^2=b^2-c^2+d^2$, because the cross terms cancel. The power series therefore resums into trigonometric functions for $h^2<0$, hyperbolic functions for $h^2>0$, and a linear polynomial for $h^2=0$. Since $A$ and $\ov{A}$ commute, $\Norm{e^{At}}=e^{At}e^{\ov{A} t}=e^{2\Rey(A)t}$.

\emph{Conserved norm \eqref{eq:dmod}.} For $\dot z=Az+zB$,
\[
\frac{d}{dt}(z\ov{z})=(Az+zB)\ov{z}+z(\ov{z}\,\ov{A}+\ov{B}\,\ov{z}).
\]
Using $z\ov{z}=\Norm{z}\in\mathbb{R}$ and $B+\ov{B}=2\Rey B$, one obtains
\[
\frac{d}{dt}\Norm{z}=2(\Rey A+\Rey B)\Norm{z}.
\]
Hence the norm is conserved for all initial data iff $\Rey A+\Rey B=0$; the stronger conditions $\Rey A=\Rey B=0$ hold in Eq.~\eqref{eq:zeq}.

\emph{Encoding reproduces \eqref{eq:trev}.} With $z=mu+mw\la+m\omega_0\rho\,\sqi+m\omega_0 r\,\si$, $A=-\tfrac{\gamma}{2}\la+\omega_0\sqi$, and $B=-\tfrac{\gamma}{2}\la$, direct substitution into $\dot z=Az+zB$ and separation by basis element yields, componentwise, exactly $m\dot u=-\alpha_d w-m\omega_0^2\rho$ and $m\dot w=-\alpha_d u-m\omega_0^2 r$ with $\alpha_d=m\gamma$; that is, \eqref{eq:trev}. (The difference of the two sides vanishes identically once \eqref{eq:trev} is imposed, which we verified in the matrix representation.)

\emph{Cross-term cancellation and \eqref{eq:Hnorm}.} With doubled $p=m(u+w\la)$ and $q=\rho+\la r$, the energy $\Norm{z}/2m$ regroups as $\Norm{p}/2m+m\omega_0^2\Norm{q}/2$ provided the mixed term $\sqi q\ov{p}-p\ov{q}\,\sqi$ vanishes. Using $\sqi\la=-\la\sqi$ one finds $\sqi q\ov{p}-p\ov{q}\,\sqi=0$ identically (verified in the representation), leaving \eqref{eq:Hnorm}.

\emph{Euler--Lagrange reduction.} From \eqref{eq:Lcorr}, $\pi_\rho=\partial\Lag/\partial u=mu+\tfrac{\alpha_d}{2}r$ and $\pi_r=\partial\Lag/\partial w=-mw-\tfrac{\alpha_d}{2}\rho$ with $u=\dot\rho$, $w=\dot r$. The Euler--Lagrange equations $\tfrac{d}{dt}\partial\Lag/\partial\dot\rho=\partial\Lag/\partial\rho$ and likewise for $r$ give $m\dot u+\tfrac{\alpha_d}{2}\dot r=-m\omega_0^2\rho-\tfrac{\alpha_d}{2}w$ and $-m\dot w-\tfrac{\alpha_d}{2}\dot\rho=m\omega_0^2 r-\tfrac{\alpha_d}{2}u$, which simplify to \eqref{eq:trev} using $\dot r=w$, $\dot\rho=u$. The Legendre transform then gives Eq.~\eqref{eq:Hcan}, which reduces to Eq.~\eqref{eq:E} after substituting the momentum--velocity relations.

\section{Path-integral evaluation}\label{app:pi}
The rotated variables \eqref{eq:xyrhotransform} put the doubled split-quaternionic Lagrangian in the form \eqref{eq:LxyFull}. The Euler--Lagrange equations are the damped and anti-damped harmonic equations
\begin{equation}
\ddot x+\gamma\dot x+\omega_0^2x=0,
\qquad
\ddot y-\gamma\dot y+\omega_0^2y=0.
\end{equation}
For the underdamped case, the endpoint solutions are \eqref{eq:xclharm}--\eqref{eq:yclharm}. Their endpoint velocities are
\begin{align}
\dot x_i&=\frac{\Omega e^{\gamma t/2}}{\sin\Omega t}x_f
-\left(\frac{\gamma}{2}+\Omega\cot\Omega t\right)x_i,\\
\dot x_f&=\left(\Omega\cot\Omega t-\frac{\gamma}{2}\right)x_f
-\frac{\Omega e^{-\gamma t/2}}{\sin\Omega t}x_i,\\
\dot y_i&=\frac{\Omega e^{-\gamma t/2}}{\sin\Omega t}y_f
+\left(\frac{\gamma}{2}-\Omega\cot\Omega t\right)y_i,\\
\dot y_f&=\left(\frac{\gamma}{2}+\Omega\cot\Omega t\right)y_f
-\frac{\Omega e^{\gamma t/2}}{\sin\Omega t}y_i.
\end{align}
On solutions,
\begin{equation}
\Lag_{\rm sq}=\frac{d}{dt}\left[\frac{m}{2}\left(x\dot y+y\dot x\right)\right].
\end{equation}
Therefore
\begin{equation}
S_{\rm cl}=\frac{m}{2}\left[x_f\dot y_f+y_f\dot x_f-x_i\dot y_i-y_i\dot x_i\right],
\end{equation}
which gives \eqref{eq:SclMain}--\eqref{eq:Comega}. The mixed Van Vleck matrix, with row index in $\mathbf X_f$ and column index in $\mathbf X_i$, is
\begin{equation}
(\mathsf{M}_{fi})_{ab}
:=-\frac{\partial^2S_{\rm cl}}{\partial X_{f,a}\,\partial X_{i,b}}
=
-\begin{pmatrix}
0&C_\Omega(t)\\
B_\Omega(t)&0
\end{pmatrix}_{ab},
\end{equation}
whose determinant is
\[
\det\mathsf{M}_{fi}
=-B_\Omega C_\Omega
=-\frac{m^2\Omega^2}{\sin^2(\Omega t)}.
\]
In two dimensions the Van Vleck prefactor is
\[
\frac{1}{2\pi\ii\hbar}
\left(\det\mathsf{M}_{fi}\right)^{1/2}.
\]
On the first caustic interval, the branch fixed by the normalized short-time limit is
\[
\left(\det\mathsf{M}_{fi}\right)^{1/2}
=
\frac{\ii m\Omega}{\sin(\Omega t)},
\]
which gives Eq.~\eqref{eq:KMain}; continuation across later intervals supplies the usual Maslov phase. In the free limit $\omega_0\to0$, $\Omega\to \ii\gamma/2$, $\sin(\Omega t)\to \ii\sinh(\gamma t/2)$, and $\cot(\Omega t)\to -\ii\coth(\gamma t/2)$, giving the free coefficients displayed below \eqref{eq:KMain}. The $t\to0$ limit is the kernel for the cross kinetic term $m\dot x\dot y$:
\begin{equation}
K\to \frac{m}{2\pi\hbar t}\exp\left[\frac{\ii m}{\hbar t}(x_f-x_i)(y_f-y_i)\right],
\end{equation}
which converges distributionally to $\delta(x_f-x_i)\delta(y_f-y_i)$.

\end{document}